\documentclass[
final
 , nomarks
]{dmtcs-episciences}

\usepackage[utf8]{inputenc}
\usepackage{subfigure}

\usepackage[numbers]{natbib}

\usepackage{amsmath,amssymb,amsthm}
\usepackage[usenames,dvipsnames]{color}
\usepackage{tikz}
\usetikzlibrary{arrows,positioning,chains,fit,shapes,calc}
\usetikzlibrary{decorations.pathreplacing}

\definecolor{mygreen}{RGB}{40,140,120}
\definecolor{truegray}{gray}{0.60}
\definecolor{mydgreen}{RGB}{21,204,0}
\definecolor{mauve}{RGB}{51,0,102}
\definecolor{rose}{RGB}{255,0,153}
\definecolor{comms}{gray}{0.55}
\definecolor{brun}{RGB}{204,0,204}
\definecolor{moche}{RGB}{129,125,10}
\definecolor{bgray}{gray}{0.25}
\definecolor{bbgray}{gray}{0.40}

\DeclareMathOperator{\tr}{Tr}

\DeclareMathOperator{\flt}{flt}

\DeclareMathOperator{\fos}{fs}
\DeclareMathOperator{\ufs}{uf}

\newcommand{\FFF}{\mathcal{F}{}}

\newcommand{\R}{\mathbb{R}}

\newtheorem{thm}{Theorem}[section]
\newtheorem{prop}{Proposition}[section]
\newtheorem{coroll}{Corollary}[section]
\newtheorem{lem}{Lemma}[section]

\theoremstyle{definition}

\newtheorem{pbl}{Problem}
\newtheorem{expl}{Example}
\theoremstyle{remark}

\newcommand{\NPclass}{$\mathcal{NP}$}

\author{Jean Cardinal\affiliationmark{1}
  \and Jean-Paul Doignon\affiliationmark{1}
  \and Keno Merckx\affiliationmark{1}}
\title[ ]{On the shelling antimatroids of split graphs}
\affiliation{
  Universit\'{e} libre de Bruxelles, Belgium}
\received{2015-12-21}
\accepted{2017-03-11}
\revised{2017-01-24}
\begin{document}
\publicationdetails{19}{2017}{1}{7}{1349}
\maketitle
\begin{abstract}  Chordal graph shelling antimatroids have received little attention with regard to their combinatorial properties and related optimization problems, as compared to the case of poset shelling antimatroids.  Here we consider a special case of these antimatroids, namely the split graph shelling antimatroids.  We show that the feasible sets of such an antimatroid relate to some poset shelling antimatroids constructed from the graph.  We discuss a few applications, obtaining in particular a simple polynomial-time algorithm to find a maximum weight feasible set.  We also provide a simple description of the circuits and the free sets.
\end{abstract}

\section{Introduction}

The ``split graph shelling antimatroids" are particular instances of ``chordal graph shelling antimatroids". To investigate them, we make use of ``poset antimatroids" (all the terms are explained in the next subsections). Our results shed light on the structure of split graph shelling antimatroids, and yield a polynomial time algorithm to find an optimal feasible set in a split graph shelling antimatroid whose elements are weighted.

\subsection*{Antimatroids}
Antimatroids arise naturally from various kinds of shellings and searches on combinatorial objects, and appear in various contexts in mathematics and computer science. \citet{Dilworth_1940} first examined structures very close to antimatroids in terms of lattice theory. Later, \citet{Edelman_80} and \citet{Jamison_1982} studied the convex aspects of antimatroids, see also~\citet{Edelman85}.  \citet{Korte_Lovasz_Schrader_1991} considered antimatroids as a subclass of greedoids. Today, the concept of antimatroid appears in many fields of mathematics such as formal language theory (\citet{Boyd_Faigle_90}), choice theory (\citet{Koshevoy_1999}), game theory (\citet{Algaba_all_04}) and mathematical psychology (\citet{Falmagne_Doignon_LS}) among others.  We became recently aware of more works on optimization in antimatroids, in particular~\citet{Queyranne15, Queyranne16}. The concept of a convex geometry is dual to the one of an antimatroid.

A set system $(V,\FFF)$, where $V$ is a finite set of elements and $ \FFF\subseteq 2^{V}$, is an \emph{antimatroid} when
\begin{align}
&V\in \FFF, \label{VinF} \tag{AM0}\\
&\forall  F_1, F_2 \in \FFF : F_1\cup F_2 \in \FFF, \label{Ustabl} \tag{AM1} \\
&\forall F \in \FFF\setminus \{\varnothing\}, \, \exists f\in F: F\setminus \{f\} \in F. \label{accec} \tag{AM2}
\end{align}
Condition \eqref{accec} is called the \emph{accessibility property}. The \emph{feasible sets} of the antimatroid $(V,\FFF)$ are the members of $\FFF$. The  \emph{convex sets} are the complements in $V$ of the feasible sets.  We have the following relation (for definitions and proof see~\citet{Edelman_80}): A finite lattice is join-distributive if and only if it is isomorphic to the lattice of feasible sets of some antimatroid.

\subsection*{Shellings}

Antimatroids also relate to special shelling processes.  Given a feasible set $F$ in an antimatroid $(V,\FFF)$, a \textsl{shelling} of $F$ is an enumeration $f_1$, $f_2$, \dots, $f_{|F|}$ of its elements such that $\{f_1$, $f_2$, \dots, $f_k\}$ is feasible for any $k$ with $1 \leq k \leq |F|$.  In view of the accessibility property, any feasible set admits at least one shelling.  There is an axiomatization of antimatroids in terms of shellings (see for example~\citet{Korte_Lovasz_Schrader_1991}).  Many examples of antimatroids arise in a natural way from shelling processes.  The next two subsections present the cases of posets and of chordal graphs. \citet{Eppstein_14} calls \textsl{basic word} of an antimatroid $(V,\FFF)$ any shelling of the whole set $V$.  He uses the notion of basic word to extend the $1/3-2/3$-conjecture to antimatroids.

\subsection*{Poset antimatroids}
Recall that a \emph{poset} $P$ is a pair $(V,\leq)$ formed of a finite set $V$ and a binary relation $\leq$ over $V$ which is reflexive, antisymmetric, and transitive. For a poset  $(V,\leq)$ a \emph{filter} $F$ is a subset of $V$ such that for all elements $a$ in $F$ and $b$ in $V$, if $a\leq b$, then $b$ is also in $F$. The filters are also known as \emph{upper ideals},  \emph{upset} or \emph{ending sets}.  We denote the family of all filters of $P$ as $\flt(V,\leq)$. 

One particular class of antimatroids comes from shelling processes over posets by removing successively the maximum elements. Let $(V,\leq)$ be a poset, then $(V, \flt(V,\leq))$ is a \emph{poset (shelling) antimatroid}. Thus the feasible elements are the filters. The class of poset antimatroids is often considered as one of the most basic, because it arises in many different contexts. Poset antimatroids are the only antimatroids closed under intersection (\citet{Korte_Lovasz_Schrader_1991}). There exist several other characterizations for this class of antimatroids. \citet{Nakamura2002, Nakamura_2003} obtains a characterization of poset antimatroids by single-element extensions and by excluded minors. Recently, \citet{Kempner2013} introduce the poly-dimension of an antimatroid, and prove that every antimatroid of poly-dimension 2 is a poset antimatroid. They establish both graph and geometric characterizations of such antimatroids. Basic words in a poset antimatroid coincide with the linear extensions of the poset (which explains their relationship to the $1/3-2/3$-conjecture).  We also note the Representation Theorem due to~\citet{birkhoff1937}: When ordered by inclusion, the feasible sets of a poset antimatroid form a distributive lattice. Conversely, any distributive lattice is isomorphic to some poset antimatroid.

From the optimization point of view, poset antimatroids with weighted elements are easy to study. First, \citet{Picard_1976} makes a direct connection between finding a maximum (or minimum) weight feasible set in a poset antimatroid and some maximum flow problems. Second, \citet{Stanley_86} provides a complete linear description of the convex hull of the characteristic vectors of the feasible sets.

\subsection*{Chordal graph shelling antimatroids}
Another particular class of antimatroids comes from shelling processes over \emph{chordal graphs}, \textit{i.e.} graphs in which every induced cycle in the graph has at most three vertices.  For a background on chordal graphs, see \citet{Golumbic2004}. Any chordal graph has at least one \emph{simplicial} vertex, \textit{i.e.} a vertex such that its neighbors induce a clique.  For any chordal graph $G=(V,E)$, we define an antimatroid $(V,\FFF)$ in which a subset $F$ of $V$ is feasible if and only if there is some ordering $O=(f_1,\ldots, f_{|F|})$ of the elements of $F$ such that for all $j$ between $1$ and $|F|$, $f_j$ is simplicial in $G\setminus\{f_1,\ldots, f_{j-1}\}$. The antimatroid resulting from this construction is called a \emph{chordal graph (vertex) shelling antimatroid}. The ordering $O$ is called a \emph{simplicial shelling} (or sometimes \emph{perfect elimination ordering}) of $F$. If applied to some arbitrary graph, the construction we just described of feasible sets in terms of simplicial vertices gives an antimatroid exactly if the graph is chordal (any graph which has a perfect elimination ordering of its whole set of vertices is chordal), see~\citet{Farber_87} for more details.

We recall that a \emph{split graph} is a graph whose set of vertices can be partitioned into a clique and an independent set (the empty set is both independent and a clique). Here we assume that for every split graph, the partition is given and we will denote by $K$ and $I$ the clique and the independent set, respectively. Split graphs are chordal graphs, and they are the only chordal graphs to be co-chordal (\textit{i.e.} the complement of the graph is also chordal), see \citet{Golumbic2004}. Here we consider the special case of chordal graph shelling antimatroids where the graph is a split graph. These antimatroids will be called \emph{split graph (vertex) shelling antimatroids}. Split graphs are relatively well known and have a wide range of theoretical use, see for instance~\citet{Merris2003, Golumbic2004, Cheng16}.

\subsection*{Structure of the paper}
The aim of this work is to provide a better understanding of the split graph shelling antimatroids. In Section~2, we establish a new characterization of the feasible sets of split graph shelling antimatroids and discuss the connection with the poset antimatroids. In Section~3 we prove a hardness result about optimization problems on antimatroids and develop a polynomial time algorithm to find a maximum (or minimum) weight feasible set in split graph shelling antimatroids. Finally in Section~4 we use the previous results to list all the circuits  and free sets of the split graph shelling antimatroids. This work is a first step to a better understanding of a more general class: the chordal graph shelling antimatroids. 

\subsection*{Notation}
Let $V$ denote a set of elements. For $S\subseteq V$, the \emph{complement} of $S$ is $S^{\complement}=V\setminus S$. The set of all subsets of $V$ is denoted as $2^{V}$ and the symbol $\uplus$ is used for the disjoint union of two sets. Let $G=(V,E)$ be a simple graph; we write $u \sim v$ as a shortcut for $\{u,v\}\in E$, and $u \nsim v$ for $\{u,v\}\notin E$. For $V'\subseteq V$ we denote by $N(V')$ the set of vertices \emph{adjacent} to $V'$, \textit{i.e.} the vertices $w$ in $V\setminus V'$ such that $w\sim v$ for some $v$ in $V'$. We write $N(v)$ for $N(\{v\})$. We call a vertex \emph{isolated} if $N(v)=\varnothing$.

\section{The split graph shelling antimatroids}

\subsection*{Characterization of the feasible sets}
Here is a useful characterization of the feasible sets in a split graph shelling antimatroid. Example~\ref{ex:NF} provides an illustration.

\begin{prop}\label{Nfclique}
Let $G=(K\cup I, E)$ be a split graph and $(V,\FFF)$ be the split graph vertex shelling antimatroid defined on $G$. Then a subset $F$ of vertices is feasible for the antimatroid if and only if $N(F)$ induces a clique.
\end{prop}

\begin{proof}
For the necessary condition (Fig.~\ref{NE_Nfclique}), suppose we have a simplicial shelling $O=(f_1,\ldots, f_{|F|})$ of a feasible set $F$ such that $N(F)$ does not induce a clique in $G$. Then, for some vertices $v_1$ and $v_2$ in $N(F)$ we have $v_1 \nsim v_2$. Hence $\{v_1,v_2\} \not\subseteq K$, since $K$ is a clique. Assume without loss of generality that $v_1 \in I$. Let $f_j$ be the first element in $O$ such that $f_j \sim v_1$. As $v_1 \in I$, by definition of a split graph, $f_j\in K$ and $f_j$ is adjacent to all other vertices of $K$. Then $f_j$ is not adjacent to $v_2$ because $f_j$ must be simplicial in $G\setminus\{f_1,\ldots, f_{j-1}\}$, so $v_2$ must be in $I$. Now let $f_t$ be the first element of $O$ such that $f_t \sim v_2$ (notice $j\neq t$). Since $v_2 \in I$, by a completely symmetric argument, we have $f_t \in K$ and $f_t\nsim v_1$. Now a contradiction follows because, if $j>t$, the vertex $f_t$ is not simplicial in $G\setminus\{f_1,\ldots, f_{t-1}\}$, and if $t>j$ the vertex $f_j$ is not simplicial in $G\setminus\{f_1,\ldots, f_{j-1}\}$.
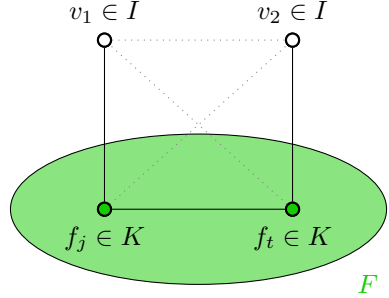
\begin{figure}[h]
\centering
\begin{tikzpicture}[scale=1]
\draw (2.25,-1) node[] {\color{mydgreen}$F$};
\draw[black,fill=mydgreen, fill opacity=0.5] (0,0) ellipse (2.5cm and 1cm);
\draw (-1.25,0) node[circle,draw,fill=gray!60,thick,inner sep=1.75pt,label=below :{$f_j \in K$}]  (t1) {};
\draw (1.25,0) node[circle,draw,fill=gray!60,thick,inner sep=1.75pt,label=below :{$f_t \in K$}]   (t2) {};
\draw (-1.25,2.25) node[circle,draw,fill=gray!60,thick,inner sep=1.75pt,label=above :{$v_1 \in I$}]  (v1) {};
\draw (1.25,2.25) node[circle,draw,fill=gray!60,thick,inner sep=1.75pt,label=above :{$v_2 \in I$}]  (v2) {};
\draw (v1)--(t1) (v2)--(t2);
\draw (t2)--(t1);
\draw node[circle,thick, fill=mydgreen,draw,inner sep=1.75pt]  at (t1) {};
\draw node[circle,thick, fill=white,draw,inner sep=1.75pt]  at (v1) {};
\draw node[circle,thick, fill=mydgreen,draw,inner sep=1.75pt]  at (t2) {};
\draw node[circle,thick, fill=white,draw,inner sep=1.75pt]  at (v2) {};
\draw[gray,dotted] (t1)--(v2);
\draw[gray,dotted] (t2)--(v1);
\draw[gray,dotted] (v1)--(v2);
\end{tikzpicture}
  \caption{Illustration of the proof of necessary condition for Proposition~\ref{Nfclique}.}
  \label{NE_Nfclique}
\end{figure}

Reciprocally, suppose that we have a set of vertices $F$ such that $N(F)$ induces a clique (Fig.~\ref{SU_Nfclique}). We will build a simplicial shelling $O$ on $F$ with the help of the following three set partition of $F$:
\begin{align*}
V_1&=F\cap I,\\
V_2&=(F\cap K) \setminus N(I\setminus F),\\
V_3&=(F\cap K) \cap N(I\setminus F).
\end{align*}
We arbitrarily order the elements in each of the sets $V_1$, $V_2$, $V_3$ and concatenate the orderings in this order to obtain the sequence $O$. By the definition of a split graph, it is obvious that the elements of $O$ in $V_1 \cup V_2$ fulfill the condition of a simplicial shelling. If $V_3 =\varnothing$, we are done. Otherwise, $N(V_3)\setminus F$ is a clique and so it has exactly one element $i$ in $I$, because $N(F)$ induces a clique and thus all elements of $V_3$ are adjacent to this single element of $I\setminus F$. Therefore the elements of $O$ in $V_3$ fulfill the conditions of the simplicial shelling. 
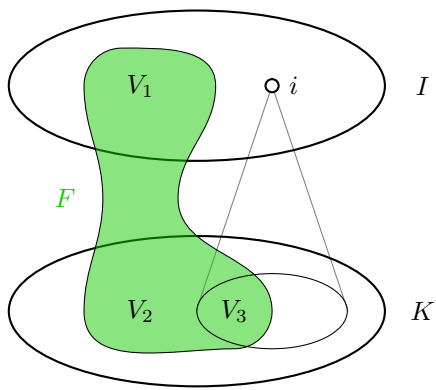
\begin{figure}[h]
\centering
 \begin{tikzpicture} 
\draw (3,0) node[] {$K$};
\draw (3,3) node[] {$I$};
 \draw[black!50] (-0.01,0)--(1,3)--(2.01,0);
  \draw[black, thick] (0,3) ellipse (2.5cm and 1cm);
  \draw[black, thick] (0,0) ellipse (2.5cm and 1cm);
  \begin{scope}[fill opacity=0.5]
    \filldraw[fill=mydgreen] ($(-1,3.5)$)
    to[out=180,in=90] ($(-1.5,3)$)
    to[out=-90,in=90] ($(-1.25,1.5)$)
    to[out=-90,in=90] ($(-1.5,0)$)
    to[out=-90,in=180] ($(0.5,-0.5)$)
    to[out=0,in=-90] ($(1,0)$)
    to[out=90,in=-90] ($(-0.25,1.5)$)
    to[out=90,in=-90] ($(0.25,3)$)
    to[out=90,in=0] ($(-1,3.5)$);
    \end{scope}
    \draw[black] (1,0) ellipse (1cm and 0.5cm);
    \draw node[circle,thick, fill=white,draw,inner sep=1.75pt,label=right :{$i$}]  at (1,3) {};
    \draw (-1.75,1.5) node[] {{\color{mydgreen}$F$}};
    \draw (-0.75,3) node[] {$V_1$};
    \draw (-0.75,0) node[] {$V_2$};
    \draw (0.5,0) node[] {$V_3$};
\end{tikzpicture}
  \caption{Illustration of the proof of the sufficient condition for Proposition~\ref{Nfclique}.}
  \label{SU_Nfclique}
\end{figure}
\end{proof}

\begin{expl}\label{ex:NF}
Figure~\ref{fig:expNF} below  shows two split graphs on which we build a split graph shelling antimatroid. The set $F$ on the left (Fig.~\ref{fig:expNF1}) is a feasible set and we see that $N(F)$ defines a clique. On the right (Fig.~\ref{fig:expNF2}), we have a clique $C$ and a possible simplicial shelling is proposed for a set of vertices such that its neighborhood is $C$. 
\end{expl}
 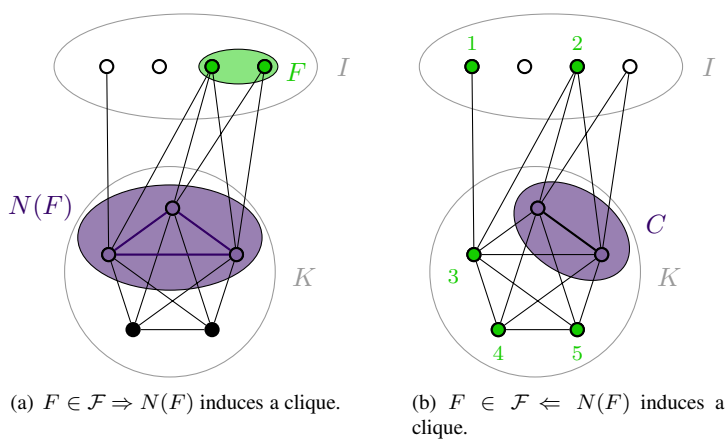
\begin{figure}[h]
        \centering
       \subfigure[$F\in\FFF$ $\Rightarrow N(F)$ induces a clique.]{ \label{fig:expNF1}
                \begin{tikzpicture}[scale=0.7]
\draw (3.25,1) node[] {{\color{truegray}$K$}};
\draw (4,5) node[] {{\color{truegray}$I$}};
\draw[black, fill=mydgreen, fill opacity=0.5] (2,5) ellipse (0.75cm and 0.33cm);
\draw[black,fill=mauve, fill opacity=0.5] (0.7,1.75) ellipse (1.75cm and 1cm);
\draw (3.1,4.9) node[] {{\color{mydgreen}$F$}};
\draw (-1.75,2.33) node[] {{\color{mauve}$N(F)$}};
\draw (0,0) node[circle,draw,thick,fill=black,black,inner sep=1.75pt]  (0) {};
\draw (0)
        ++ (0:1.5cm) node[circle,draw,thick,fill=black,black,inner sep=1.75pt] (d)  {}
        ++ (72:1.5cm) node[circle,draw,thick,fill=black,black,inner sep=1.75pt] (c)  {}
        ++ (144:1.5cm) node[circle,draw,thick,fill=black,black,inner sep=1.75pt] (b)  {}
        ++ (216:1.5cm) node[circle,draw,thick,fill=black,black,inner sep=1.75pt] (a)  {};
 \draw[truegray] (0.7,1.1) ellipse (2cm and 2cm);
 \draw (-0.5,5) node[circle,draw,fill=white,thick,inner sep=1.75pt] (1)  {}
 (0.5,5) node[circle,draw,fill=white,thick,inner sep=1.75pt] (2)  {}
 (1.5,5) node[circle,draw,fill=white,thick,inner sep=1.75pt] (3)  {}
 (2.5,5) node[circle,draw,fill=white,thick,inner sep=1.75pt] (4)  {};
  \draw[truegray] (1,5) ellipse (2.5cm and 1cm);with
 \draw node[circle,thick, fill=mydgreen,draw,inner sep=1.75pt]  at (4) {};
 \draw node[circle,thick,fill=mydgreen,draw,inner sep=1.75pt]  at (3) {};
 \draw node[circle,thick,fill=mauve!60,draw,inner sep=1.75pt]  at (a) {};
 \draw node[circle,thick,fill=mauve!60,draw,inner sep=1.75pt]  at (b) {};
 \draw node[circle,thick,fill=mauve!60,draw,inner sep=1.75pt]  at (c) {};
\draw (a)--(b)--(c)--(d)--(0)--(a)--(d)--(b)--(0)--(c)--(a);  
\draw (c)--(4)--(b)--(3) (a)--(1) (c)--(3)--(a);
\draw[mauve,thick] (b)--(c)--(a)--(b);
\end{tikzpicture}}\qquad
\subfigure[$F\in\FFF$ $\Leftarrow N(F)$ induces a clique.]{ \label{fig:expNF2}
                \begin{tikzpicture}[scale=0.7]
\draw (3.25,1) node[] {{\color{truegray}$K$}};
\draw (4,5) node[] {{\color{truegray}$I$}};
\draw (3,2) node[] {{\color{mauve}$C$}};
\draw[rotate=-33,black,fill=mauve, fill opacity=0.5] (0.15,2.33) ellipse (1.2cm and 0.8cm);
\draw (0,0) node[circle,draw,thick,fill=black,black,inner sep=1.75pt]  (0) {};
\draw (0)
        ++ (0:1.5cm) node[circle,draw,thick,fill=black,black,inner sep=1.75pt] (d)  {}
        ++ (72:1.5cm) node[circle,draw,thick,fill=black,black,inner sep=1.75pt] (c)  {}
        ++ (144:1.5cm) node[circle,draw,thick,fill=black,black,inner sep=1.75pt] (b)  {}
        ++ (216:1.5cm) node[circle,draw,thick,fill=black,black,inner sep=1.75pt] (a)  {};
 \draw[truegray] (0.7,1.1) ellipse (2cm and 2cm);
\draw (-0.5,5) node[circle,draw,fill=white,thick,inner sep=1.75pt] (1)  {}
 (0.5,5) node[circle,draw,fill=white,thick,inner sep=1.75pt] (2)  {}
 (1.5,5) node[circle,draw,fill=white,thick,inner sep=1.75pt] (3)  {}
 (2.5,5) node[circle,draw,fill=white,thick,inner sep=1.75pt] (4)  {};
  \draw[truegray] (1,5) ellipse (2.5cm and 1cm);
 \draw node[circle,thick, fill=mauve!60,draw,inner sep=1.75pt]  at (c) {};
 \draw node[circle,thick, fill=mauve!60,draw,inner sep=1.75pt]  at (b) {};
 \draw node[circle,thick, fill=mydgreen,draw,inner sep=1.75pt,label=below:{\footnotesize \color{mydgreen} $4$}]  at (0) {};
 \draw node[circle,thick,fill=mydgreen,draw,inner sep=1.75pt,label=above :{\footnotesize \color{mydgreen} $1$}]  at (1) {};
 \draw node[circle,thick, fill=mydgreen,draw,inner sep=1.75pt,label=above :{\footnotesize \color{mydgreen} $2$}] at (3) {};
 \draw node[circle,thick,fill=mydgreen,draw,inner sep=1.75pt,label=below left :{\footnotesize \color{mydgreen} $3$}]  at (a) {};
 \draw node[circle,thick,fill=mydgreen,draw,inner sep=1.75pt,label=below :{\footnotesize \color{mydgreen} $5$}]  at (d) {};
\draw (a)--(b)--(c)--(d)--(0)--(a)--(d)--(b)--(0)--(c)--(a);weight 
\draw (c)--(4)--(b)--(3) (a)--(1) (c)--(3)--(a);
\draw[black,thick]  (c)--(b);
\end{tikzpicture}}
        \caption{Examples for Proposition~\ref{Nfclique}.}\label{fig:expNF}
\end{figure}

We recall that a \emph{chordless path} in a graph is a path for which no two vertices are connected by an edge that is not in the path. The chordless paths are also called \emph{induced paths}. In a graph $(V,E)$, a subset $C$ of $V$ is \textsl{monophonically convex} (\textsl{m-convex}) when $C$ contains all the vertices of all chordless paths  joining any two vertices of $C$.

\begin{prop}\label{prop_1}
Let $(V,E)$ be a graph and $F$ be a subset of $V$. If $N(F)$ is a clique, then $V \setminus F$ is m-convex.
\end{prop}

\begin{proof}
Assume $N(F)$ is a clique.  Proceeding by contradiction, we take two vertices $v$, $w$ in $V \setminus F$ for which there exists a chordless path $v=u_0, u_1, \dots, u_k=w$ having at least one vertex in $F$. Now we select $i$ minimal and $j$ maximal in $\{1,\ldots,k-1\}$ such that $u_i$ and $u_j$ are in $F$.  Then necessarily $u_{i-1}$ and $u_{j+1}$ are adjacent, so the path has chord, contradiction.
\end{proof}

The converse of the implication in Proposition~\ref{prop_1} does not hold even if the graph is connected. Even more:  $V\setminus F$ being m-convex does not imply that the graph $N$ induced on $N(F)$ is a parallel sum of cliques (in other words, that $N$ is the complement of a multipartite graph). Figure~\ref{expJPD} below shows a counter-example based on a 2-connected, chordal graph.

\begin{figure}[h]
\centering
\begin{tikzpicture}[scale=1]
\tikzstyle{vertex}=[circle,draw,fill=white,thick,inner sep=1.75pt]
\node[vertex,fill=mydgreen] (a) at (-1,0) {};
\node[vertex,fill=mydgreen] (b) at (0,0) {};
\node[vertex,fill=mydgreen] (c) at (1,0) {};
\node[vertex] (d) at (-1,1) {};
\node[vertex] (e) at (0,1) {};
\node[vertex] (f) at (1,1) {};
\node[vertex] (g) at (-0.5,2) {};
\node[vertex] (h) at (0.5,2) {};
\node[vertex] (i) at (1.5,2) {};
\draw (d) -- (a) -- (g) -- (d) -- (e) -- (b) -- (h) -- (e) -- (f) -- (c) -- (i) -- (f);
\draw (g) -- (h) -- (i);
\draw (g) -- (e) -- (i);
\draw (2.7,0) node[] {${\color{mydgreen}F}$};
\draw[black,fill= mydgreen, fill opacity=0.5] (0,0) ellipse (2cm and 0.75cm);
\end{tikzpicture}
      \caption{Counter-example for the converse implication in Proposition~\ref{prop_1}.}\
       \label{expJPD}
\end{figure}
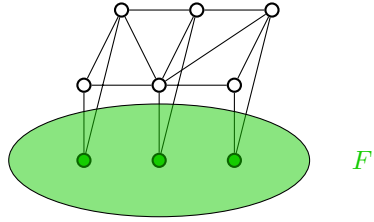

For a split graph $(V,E)$, the converse of the implication in Proposition~\ref{prop_1} also holds (this follows from Proposition~\ref{Nfclique} and Section~3 of~\citet{Farber_87}).

\begin{coroll} \label{Lem_i1i2}
Let $G=(K\cup I, E)$ be a split graph and $(V,\FFF)$ be the split graph vertex shelling antimatroid defined on $G$. For all feasible sets $F$, there is at most one $i \in I \setminus F$ such that there is $k \in K \cap F$ with $k\sim i$.
\end{coroll}

\begin{proof} 
This comes directly from Proposition~\ref{Nfclique} and the fact that the set $I$ is an independent set.
\end{proof}

\begin{coroll} \label{defGraph1}
Let $G=(K\cup I, E)$ be a split graph and $(V,\FFF)$ be the split graph vertex shelling antimatroid defined on $G$. Let $u$ and $v$ be distinct elements in $V$. Then $V\setminus \{u,v\}\in \FFF$ if and only if $u\sim v$, or at least one of the two vertices is isolated in $G$.
\end{coroll}
  
\begin{proof}
This comes directly from Proposition~\ref{Nfclique}.
\end{proof}

Corollary~\ref{defGraph1} helps us to rebuild the original split graph for a given split graph shelling antimatroid, as shown in the following proposition.

\begin{prop}\label{unicA}
Let $(V,\FFF)$ be a split graph shelling antimatroid with $\FFF \neq 2^{V}$, then there is a unique split graph $G$ such that $(V,\FFF)$ is the split graph shelling antimatroid defined on $G$.
\end{prop}

\begin{proof}
Suppose we have obtained a graph $G$ such that $(V,\FFF)$ is the split graph shelling antimatroid defined on it. Because  $\FFF \neq 2^{V}$ and Proposition~\ref{Nfclique}, the graph $G$ must have a non-empty subset $S$ of vertices such that there exist $a,b$ in $N(S)$ with $a\nsim b$  (thus $V\setminus \{a,b\} \notin \FFF$).

If we take an element $v$ in $V$ such that $V\setminus \{u,v\}\in \FFF$ for all $u$ in $V\setminus\{v\}$ (so $v\notin \{a,b\}$), then Corollary~\ref{defGraph1} leaves two options. Either the vertex $v$ is isolated in $G$ or $v$ forms an edge with every non-isolated vertex in $G$. Moreover, if this element $v$ is such that $\{v\}\in \FFF$, then the existence of the subset $S$ in the graph and Proposition~\ref{Nfclique} imply that $v$ must be an isolated vertex in the graph.

We now build the graph $G=(V,E)$ as follows. First, we identify the isolated vertices as the vertices $i$ satisfying $V\setminus \{i,u\}\in \FFF$ for all $u$ in $V\setminus\{i\}$ and $\{i\}\in \FFF$. Next, among all pairs of non-isolated vertices $\{v_1,v_2\}$, the ones that give an edge in $G$ satisfy $V\setminus \{v_1,v_2\}\in \FFF$. We know that there is no other edge by Corollary~\ref{defGraph1}.
\end{proof}

Remark that for the split graph shelling antimatroid $(V,2^{V})$, there exist several split graphs such that $(V,2^{V})$ is the split graph shelling antimatroid defined on it. For instance the complete graph on $V$, or the graph $(V,\varnothing)$.

Testing whether a given antimatroid $(V,\FFF)$ is a split graph shelling antimatroid can be done using arguments in the proof of Proposition~\ref{unicA}: First, build a graph $G=(V,E)$ with $\{v_1,v_2\}\in E$ exactly if $V\setminus \{v_1,v_2\}\in \FFF$ and $V\setminus\{v_1,u_1\}\notin \FFF$ for some $u_1$ in $V\setminus \{v_1\}$ and $V\setminus \{v_2,u_2\}\notin \FFF$ for some $u_2$ in $V\setminus\{v_2\}$. Next check that $G$ is split and $\FFF$ consists of exactly the feasible sets of $G$.

We now distinguish two classes of feasible sets for the split graph shelling antimatroids. A feasible set $F$ is an \emph{$i$-feasible set} if there exists some vertex $i$ in $N(F)\cap I$ (by Corollary~\ref{Lem_i1i2}, such an $i$ is unique). On the other hand, a feasible set $F$ is a \emph{$*$-feasible set} when $N(F)\subseteq K$. Figure~\ref{fig:FWfiltre} below illustrates the two classes of feasible sets.

\begin{figure}[h]
        \centering
       \subfigure[$F_2$ is an $i$-feasible set.]{\label{fig:Wfiltre}
\begin{tikzpicture}[scale=0.7]
\draw (4.5,0.75) node[] {{\color{truegray}$K$}};
\draw (4.5,5) node[] {{\color{truegray}$I$}};
\draw (4,2.5) node[] {{\color{mydgreen}$F_2$}};
 \draw[truegray] (1.5,0.75) ellipse (2.5cm and 1.5cm);
 \draw[truegray] (1.5,5) ellipse (2.5cm and 1cm);
\draw (0,0) node[circle,draw,thick,fill=black,black,inner sep=1.75pt]  (k2) {};
\draw (3,0) node[circle,draw,thick,fill=black,black,inner sep=1.75pt]  (k1) {};
\draw (1.5,1.5) node[circle,draw,thick,fill=black,black,inner sep=1.75pt]  (k3) {};
  \draw (3,5) node[circle,draw,fill=mydgreen,thick,inner sep=1.75pt] (i3)  {}
 (1.5,5) node[circle,draw,fill=white,thick,inner sep=1.75pt, label=above :{$i$}] (i2)  {}
 (0,5) node[circle,draw,fill=white,thick,inner sep=1.75pt] (i1)  {};
\draw (k1)--(k2)--(k3)--(k1);
\draw (i1)--(k2)--(i2)--(k1)--(i3);
   \begin{scope}[fill opacity=0.5]
    \filldraw[fill=mydgreen] ($(k1)+(0.5,0)$)
    to[out=90,in=-90] ($(i3) + (0.5,0)$)
    to[out=90,in=90] ($(i3) - (0.5,0)$)
    to[out=-90,in=90] ($(k3) - (0.5,0)$)
    to[out=-90,in=90] ($(k1)-(0.5,0)$)
    to[out=-90,in=-90] ($(k1)+(0.5,0)$);
    \end{scope}
\end{tikzpicture}}\qquad
\subfigure[$F_1$ is a $*$-feasible set.]{\label{fig:Ffiltre}
                 \begin{tikzpicture}[scale=0.7]
\draw (4.5,0.75) node[] {{\color{truegray}$K$}};
\draw (4.5,5) node[] {{\color{truegray}$I$}};
\draw (-1,2.5) node[] {{\color{mydgreen}$F_1$}};
 \draw[truegray] (1.5,0.75) ellipse (2.5cm and 1.5cm);
 \draw[truegray] (1.5,5) ellipse (2.5cm and 1cm);
\draw (0,0) node[circle,draw,thick,fill=black,black,inner sep=1.75pt]  (k2) {};
\draw (3,0) node[circle,draw,thick,fill=black,black,inner sep=1.75pt]  (k1) {};
\draw (1.5,1.5) node[circle,draw,thick,fill=black,black,inner sep=1.75pt]  (k3) {};
  \draw (3,5) node[circle,draw,fill=white,thick,inner sep=1.75pt] (i3)  {}
 (1.5,5) node[circle,draw,fill=mydgreen,thick,inner sep=1.75pt] (i2)  {}
 (0,5) node[circle,draw,fill=mydgreen,thick,inner sep=1.75pt] (i1)  {};
\draw (k1)--(k2)--(k3)--(k1);
\draw (i1)--(k2)--(i2)--(k1)--(i3);
   \begin{scope}[fill opacity=0.5]
    \filldraw[fill=mydgreen] ($(k2)-(0.5,0)$)
    to[out=90,in=-90] ($(i1) - (0.5,0)$)
    to[out=90,in=90] ($(i2) + (0.5,0)$)
    to[out=-90,in=90] ($(k2) + (0.5,0)$)
    to[out=-90,in=-90] ($(k2)-(0.5,0)$);
    \end{scope}
\end{tikzpicture}}
        \caption{Examples of the two classes of feasible sets.}\label{fig:FWfiltre}
\end{figure}
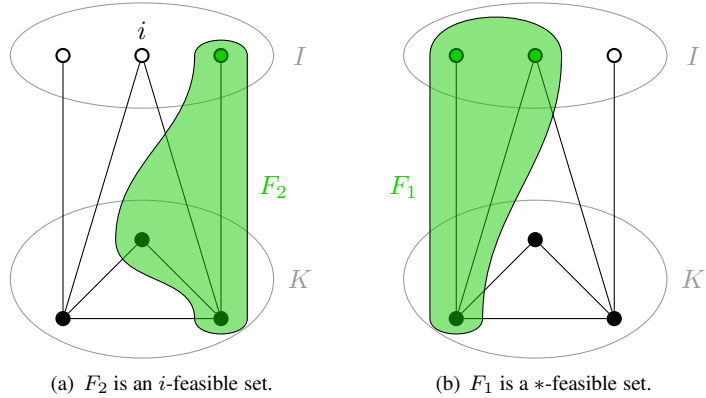

A feasible set of a split graph shelling antimatroid belongs either to the family $\FFF_{*}$ of $*$-feasible sets, or to one family $\FFF_{i}$ of $i$-feasible sets as shown in the following corollary. The proof is straightforward.

\begin{coroll} \label{partitionF}
Let $G=(K\cup I, E)$ be a split graph and $(V,\FFF)$ be the split graph vertex shelling antimatroid defined on $G$. If $I=\{i_1,\ldots,i_{|I|}\}$, then $\FFF$ decomposes into 
\[\FFF_{*} \uplus \FFF_{i_1} \uplus \cdots \uplus \FFF_{|I|}.\]
\end{coroll}


\subsection*{Connection between split graph shellings and poset antimatroids}
For investigating a split graph $(K\cup I, E)$, we will make use of two functions from $I$ to $2^{K\cup I}$, the \emph{forced set function} and the \emph{unforced set function}, respectively:
\begin{align*}
\fos(i)=&\{k\in K : k\nsim i\} \cup \{i'\in I : N(i') \not\subseteq N(i)\},\\
\ufs(i)=&\fos(i)^{\complement}\setminus \{i\}=\{k\in K : k\sim i\} \cup \{i'\in I : N(i') \subseteq N(i)\}\setminus \{i\} .
\end{align*}
As shown in the next lemma, the forced set function evaluated at $i$ gives us the vertices which belong in any $i$-feasible set. The unforced set function evaluated at $i$ just gives the complement of $\fos(i)$, minus $i$. So for all $i$ in $I$ the vertex set of the graph is equal to $\{i\}\cup \fos(i) \cup \ufs(i)$. Those two definitions are illustrated in Figure~\ref{fig:prices}.

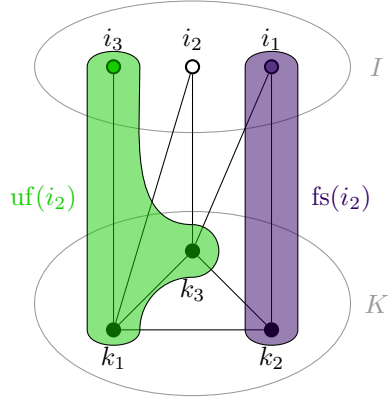
\begin{figure}[h]
        \centering
\begin{tikzpicture}[scale=0.7]
\draw (5,0.5) node[] {{\color{truegray}$K$}};
\draw (5,5) node[] {{\color{truegray}$I$}};
\draw (-1.33,2.5) node[] {{\color{mydgreen}$\ufs(i_2)$}};
\draw (4.33,2.5) node[] {{\color{mauve}$\fos(i_2)$}};
\draw[truegray] (1.5,0.5) ellipse (3cm and 1.75cm);
\draw[truegray] (1.5,5) ellipse (3cm and 1.25cm);
\draw (0,0) node[circle,draw,thick,fill=black,black,inner sep=1.75pt,label=below:{$k_1$}]  (k2) {};
\draw (3,0) node[circle,draw,thick,fill=black,black,inner sep=1.75pt,label=below :{$k_2$}]  (k1) {};
\draw (1.5,1.5) node[circle,draw,thick,fill=black,black,inner sep=1.75pt,label={[label distance=0.15cm]-90:$k_3$}]  (k3) {};
  \draw (3,5) node[circle,draw,fill=mauve!60,thick,inner sep=1.75pt,label=above :{$i_1$}] (i3)  {}
 (1.5,5) node[circle,draw,fill=white,thick,inner sep=1.75pt,label=above :{$i_2$}] (i2)  {}
 (0,5) node[circle,draw,fill=mydgreen,thick,inner sep=1.75pt,label=above :{$i_3$}] (i1)  {};
\draw (k1)--(k2)--(k3)--(k1);
\draw (i1)--(k2)--(i2)--(k3)--(i3)--(k1);
   \begin{scope}[fill opacity=0.5]
    \filldraw[fill=mydgreen] ($(k2)-(0.5,0)$)
    to[out=90,in=-90] ($(i1) - (0.5,0)$)
    to[out=90,in=90] ($(i1) + (0.5,0)$)
    to[out=-90,in=180] ($(k3) + (0,0.5)$)
    to[out=0,in=90] ($(k3) + (0.5,0)$)
    to[out=-90,in=0] ($(k3) - (0,0.5)$)
    to[out=180,in=90] ($(k2) + (0.5,0)$)
    to[out=-90,in=-90] ($(k2)-(0.5,0)$);
    \end{scope}
     \begin{scope}[fill opacity=0.5]
    \filldraw[fill=mauve] ($(k1)+(0.5,0)$)
    to[out=90,in=-90] ($(i3) + (0.5,0)$)
    to[out=90,in=90] ($(i3) - (0.5,0)$)
    to[out=-90,in=90] ($(k1)-(0.5,0)$)
    to[out=-90,in=-90] ($(k1)+(0.5,0)$);
    \end{scope}
\end{tikzpicture}
        \caption{Illustration of the forced set and unforced set functions.}\label{fig:prices}
\end{figure}

\begin{lem} \label{unblIN}
Let $G=(K\cup I, E)$ be a split graph and $(V,\FFF)$ be the split graph vertex shelling antimatroid defined on $G$. Let $i$ be in $I$, then for any $i$-feasible set $F$ in $\FFF$ we have $\fos(i)\subseteq F$.
\end{lem}

\begin{proof} 
If $F$ is an $i$-feasible set, we have $i\in I\cap F^{\complement}$ and there is a $k$ in $K\cap F$ with $k\sim i$. If a vertex $v$ in $\fos(i)$ is not in $F$ then we have two possibilities. Either $v\in K$ and so $i\nsim v$ (by definition of $\fos(i)$), but that contradicts Proposition~\ref{Nfclique} because $\{i,v\}\subseteq N(F)$, or $v\in I$ and there is $k'\in K$ such that $k'\nsim i$ but $k'\sim v$ (by definition of $\fos(i)$). We know that $k\sim k'$ by definition of $K$, but that also contradicts Proposition~\ref{Nfclique} because if $k'\notin F$, then $\{i,k'\}\subseteq N(F)$, and if $k'\in F$ then $\{i,v\}\subseteq N(F)$ with $i\nsim v$ by definition of $I$.
\end{proof}

We will now establish the link between split graph shelling antimatroids and poset antimatroids. If we have a split graph $G=(K\cup I, E)$, we build a poset on $K\cup I$ with the binary relation $\prec$ defined by $u \prec v$  if and only if $u\in K$, $v\in I$ and $u\sim v$ in $G$. The resulting poset $(K\cup I,\prec)$ is of height at most two (the number of elements in a chain is at most two). Next, we prove that all the structures $(V,\FFF_{*})$ and $(\ufs(i),\{F\setminus \fos(i) : F \in \FFF_{i}\}\cup \varnothing)$ for $i$ in $I$ are poset antimatroids.

\begin{prop}\label{pAMposet0.0}
Let $G=(K\cup I, E)$ be a split graph and $(V,\FFF)$ be the split graph vertex shelling antimatroid defined on $G$, then $\FFF_{*}=\flt(K\cup I,\prec)$.
\end{prop}

\begin{proof}
First we show that  $\FFF_{*}\subseteq \flt(K\cup I,\prec)$. Take  $F$ in $\FFF_{*}$, by definition of a $*$-feasible set, if there is a $k\in F\cap K$, then $N(k) \cap I \subseteq F$. Then $F$ is a filter of $(K\cup I,\prec)$. 

Next we show that  $\FFF_{*}\supseteq \flt(K\cup I,\prec)$. Suppose that $F$ is  a filter of  $(K\cup I,\prec)$, then $N(F) \subseteq K$ by the definition of $\prec$, and by Proposition~\ref{Nfclique} we know that $F$ is a feasible set because $K$ is a clique, and also a $*$-feasible set because $N(F)\cap I =\varnothing$.
\end{proof}

In the following, we use $(\ufs(i),\prec)$ to denote the poset formed on $\ufs(i)$ with the binary relation $\prec$ restricted to $\ufs(i)$.

\begin{prop}\label{pAMposet0.1}
Let $G=(K\cup I, E)$ be a split graph and $(V,\FFF)$ be the split graph vertex shelling antimatroid defined on $G$, then for all $i$ in $I$, $\FFF_{i}=\{\fos(i)\cup H: H\in \flt(\ufs(i),\prec)\, , H\cap K \neq \varnothing \}$.
\end{prop}

\begin{proof}
Let $i$ be in $I$. We first show that  $\FFF_{i}\subseteq \{\fos(i)\cup H: H\in \flt(\ufs(i),\prec)\, , H\cap K \neq \varnothing \}$. Take a $F$ in $\FFF_{i}$, then $\fos(i)\subseteq F$ by Lemma~\ref{unblIN}. We have directly from the definition of $\FFF_i$ that $(F\setminus \fos(i))  \cap K \neq \varnothing$, note also that $F\setminus \fos(i)=F\cap \ufs(i)$. Now we just want to show that $F\cap \ufs(i)$ is a filter of $(\ufs(i),\prec)$. This is equivalent to showing that $N(k) \cap I \cap \ufs(i) \subseteq F$ for all $k$ in $K\cap (F\cap \ufs(i))$. So if we take a  $k$ in $K\cap (F\cap \ufs(i))$, then $k\sim i$, but $F$ is a feasible set, so we must have, by Corollary~\ref{Lem_i1i2}, $N(k)\cap I \cap \ufs(i) \subseteq F$. 

Secondly, we show that $\FFF_{i}\supseteq \{\fos(i)\cup H: H\in \flt(\ufs(i),\prec)\, , H\cap K \neq \varnothing \}$. Suppose that $H$ is a filter of  $(\ufs(i),\prec)$ such that $H\cap K \neq \varnothing $. We need to show that $H\cup \fos(i)$ is a feasible set. We use again Proposition~\ref{Nfclique} and check that $N(\fos(i) \cup H)$ induces a clique. We only need to observe that $N(\fos(i))\subseteq  N(i)$ by definition of the function $\fos$, and this implies $N(\fos(i) \cup H)\subseteq N(i)\cup\{i\}$ which is a clique. Finally, by construction $H\cap N(i)\neq  \varnothing$ and $i \notin H$, so $\fos(i) \cup H$ is an $i$-feasible set and the proof is complete.
\end{proof}

\begin{coroll}
Let $G=(K\cup I, E)$ be a split graph and $(V,\FFF)$ be the split graph vertex shelling antimatroid defined on $G$, then $(V,\FFF_{*})$ and $(\ufs(i),\{F\setminus \fos(i) : F \in \FFF_{i}\}\cup \varnothing)$ for $i$ in $I$ are all poset antimatroids.
\end{coroll}

\begin{proof}
This follows directly from Propositions~\ref{pAMposet0.0} and~\ref{pAMposet0.1}.
\end{proof}

Proposition~\ref{pAMposet0.1} shows us that an $i$-feasible set can be decomposed into $\fos(i)$ and a filter of $(\ufs(i),\prec)$. It is easy to see that this decomposition is unique.

The above definitions of $*$-feasible and $i$-feasible sets lead to a better understanding of the poset produced by the feasible sets of a split graph shelling antimatroid. Indeed, for a split graph shelling antimatroid $(V,\FFF)$ built on a split graph $(K\cup I, E)$, we decompose the structure of the poset formed by its feasible sets into a poset $(\FFF_{*},\subseteq)$ and $|I|$ posets $(\FFF_{i},\subseteq)$, for $i\in I$, as illustrated by Figure~\ref{LattF} and detailed in Example~\ref{LattFex}.

\begin{figure}[h]
        \centering
\begin{tikzpicture}[scale=0.7]
\draw (0,0) node[fill=white, inner sep=2pt] (empt) {$\varnothing$};
  \begin{scope}[fill opacity=0.5]
    \filldraw[fill=mydgreen] (empt)--(3,4)[dotted]--(-3,4)--(empt);
    \end{scope}
 \draw (-3,4)--(empt)--(3,4);
 \draw (2,1) node[] {{\color{mydgreen}$(\FFF_{*},\subseteq)$}};
 \draw  (-0.75,1.5) node[circle,draw,fill=black,thick,inner sep=1.25pt,label=right :{\footnotesize $\fos(i)$}] (fsi)  {};
  \draw  (0.5,2.5) node[circle,draw,fill=black,thick,inner sep=1.25pt,label=left :{\footnotesize $\fos(i')$}] (fsi2)  {};
 \draw[dashed] (fsi)--(-4.5,1.5)--(-4.5,2);
  \begin{scope}[fill opacity=0.5]
    \filldraw[fill=mauve] (-4.75,2)--(-5.25,3.5)[dotted]--(-3.75,3.5)--(-4.25,2);
    \end{scope}  
  \draw (-5.25,3.5)--(-4.75,2)--(-4.25,2)--(-3.75,3.5);
 \draw[dashed] (fsi2)--(5.75,2.5)--(5.75,3);  
  \begin{scope}[fill opacity=0.5]
    \filldraw[fill=mauve] (6,3)--(6.5,4.5)[dotted]--(5,4.5)--(5.5,3);
    \end{scope}  
  \draw (6.5,4.5)--(6,3)--(5.5,3)--(5,4.5);
 \draw  (-6,2.5) node[] {{\color{mauve}$(\FFF_{i},\subseteq)$}};
 \draw  (7.25,3.5) node[] {{\color{mauve}$(\FFF_{i'},\subseteq)$}};
\end{tikzpicture}
  \caption{Schematic view of the poset produced by the feasible sets.}
  \label{LattF}
\end{figure}
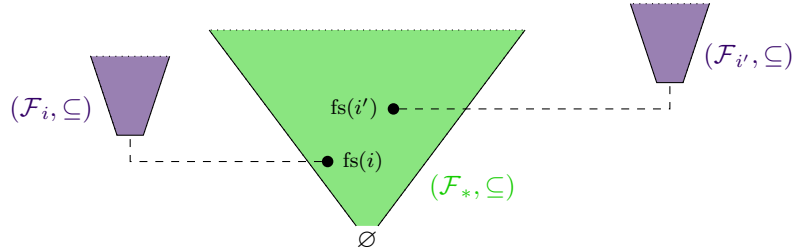

 \begin{expl} \label{LattFex}
Figure~\ref{ex:LattFtot} shows a split graph $(K\cup I,E)$ and the poset formed  by the feasible sets of the  split graph shelling antimatroid built on it. The dashed link between a feasible set $F$ and the union sign means that above this point, we look at $i$-feasible sets (for some $i \in I$) and the sets considered must be taken in union with $\fos(i)$. 
 \end{expl}

   \begin{figure}[h]
        \centering
\begin{tikzpicture}[scale=0.7]
\draw (3.5,0) node[] {{\color{truegray}$K$}};
\draw (3,3) node[] {{\color{truegray}$I$}};
 \draw[truegray] (1,0) ellipse (2cm and 1cm);
 \draw[truegray] (0.5,3) ellipse (2cm and 1cm);
 \draw (0,0) node[circle,draw,fill=black,thick,inner sep=1.75pt,label=below :{\footnotesize $1$}] (1)  {}
 (1,0.75) node[circle,draw,fill=black,thick,inner sep=1.75pt,label=below :{\footnotesize $2$}] (2)  {}
 (2,0) node[circle,draw,fill=black,thick,inner sep=1.75pt,label=below :{\footnotesize $3$}] (3)  {};
 \draw (-0.5,3) node[circle,draw,fill=white,thick,inner sep=1.75pt,label=above :{\footnotesize $4$}] (4)  {}
 (0.5,3) node[circle,draw,fill=white,thick,inner sep=1.75pt,label=above :{\footnotesize $5$}] (5)  {}
 (1.5,3) node[circle,draw,fill=white,thick,inner sep=1.75pt,label=above :{\footnotesize $6$}] (6)  {};
\draw (4)--(1)--(5)--(2)--(6)--(3);
\draw (1)--(2)--(3)--(1);
\end{tikzpicture}

\bigskip

        \begin{tikzpicture}[scale=0.7]
\draw (0,0) node[fill=white, inner sep=0.5pt] (empt) {\footnotesize $\varnothing$}  
      (-3,1.5) node[fill=white, inner sep=0.5pt] (4) {\footnotesize $\{4\}$}
      (0,1.5) node[fill=white, inner sep=0.5pt] (5) {\footnotesize $\{5\}$}
      (3,1.5) node[fill=white, inner sep=0.5pt] (6) {\footnotesize $\{6\}$}
      (-4.5,3) node[fill=white, inner sep=0.5pt] (45) {\footnotesize $\{4,5\}$}
      (-1.5,3) node[fill=white, inner sep=0.5pt] (56) {\footnotesize $\{5,6\}$}
      (1.5,3) node[fill=white, inner sep=0.5pt] (46) {\footnotesize $\{4,6\}$}
      (4.5,3) node[fill=white, inner sep=0.5pt] (36) {\footnotesize $\{3,6\}$}
      (-6,4.5) node[fill=white, inner sep=0.5pt] (145) {\footnotesize $\{1,4,5\}$}
      (-3.5,4.5) node[fill=white, inner sep=0.5pt] (256) {\footnotesize $\{2,5,6\}$}
      (0,4.5) node[fill=white, inner sep=0.5pt] (456) {\footnotesize $\{4,5,6\}$}
      (3.5,4.5) node[fill=white, inner sep=0.5pt] (356) {\footnotesize $\{3,5,6\}$}
      (6,4.5) node[fill=white, inner sep=0.5pt] (346) {\footnotesize $\{3,4,6\}$}
      (-4.5,6) node[fill=white, inner sep=0.5pt] (1456) {\footnotesize $\{1,4,5,6\}$}
      (-1.5,6) node[fill=white, inner sep=0.5pt] (2456) {\footnotesize $\{2,4,5,6\}$}
      (1.5,6) node[fill=white, inner sep=0.5pt] (3456) {\footnotesize $\{3,4,5,6\}$} 
      (4.5,6) node[fill=white, inner sep=0.5pt] (2356) {\footnotesize $\{2,3,5,6\}$}
      (-3,7.5) node[fill=white, inner sep=0.5pt] (12456) {\footnotesize $\{1,2,4,5,6\}$}
      (0,7.5) node[fill=white, inner sep=0.5pt] (23456) {\footnotesize $\{2,3,4,5,6\}$}
      (3,7.5) node[fill=white, inner sep=0.5pt] (13456) {\footnotesize $\{1,3,4,5,6\}$}
      (0,9) node[fill=white, inner sep=0.5pt] (V) {\footnotesize $\{1,2,3,4,5,6\}$};
\draw (empt)--(4) (empt)--(5) (empt)--(6);
\draw (4)--(45)--(5)--(56)--(6)--(46)--(4) (6)--(36);
\draw (145)--(45) (256)--(56)--(356)--(36) (46)--(346)--(36) (56)--(456)--(45) (456)--(46);
\draw (145)--(1456) (2456)--(256)--(2356)--(356)--(3456)--(346) (2456)--(456)--(1456) (3456)--(456);
\draw (1456)--(12456)--(2456)--(23456)--(2356) (1456)--(13456)--(3456);
\draw (V)--(12456) (V)--(23456) (V)--(13456);
\draw[decorate,thick,mydgreen,decoration={brace,amplitude=5pt,mirror}] 
    (-6,-0.25)  -- (6,-0.25)  node [midway,yshift=-0.33cm] {\scriptsize \color{mydgreen} $(\FFF_{*},\subseteq)$};
\node (5F)  at (9,3) {$  \cup$};
\draw (10,4.5) node[fill=white, inner sep=0.5pt] (5F2) {\footnotesize $\{2\}$}
      (10,6) node[fill=white, inner sep=0.5pt] (5F24) {\footnotesize $\{2,4\}$}
      (8,6) node[fill=white, inner sep=0.5pt] (5F14)  {\footnotesize $\{1,4\}$}
      (9,7.5) node[fill=white, inner sep=0.5pt] (5F124)  {\footnotesize $\{1,2,4\}$};
\draw[dashed] (36)--(5F) ;
\draw[dashed] (5F2)--(5F)--(5F14);
\draw (5F24)--(5F2);
\draw (5F14)--(5F124)--(5F24);
\draw[decorate,thick,mauve,decoration={brace,amplitude=5pt,mirror}] 
    (8,2.75)  -- (10,2.75)  node [midway,yshift=-0.33cm] {\scriptsize \color{mauve} $(\FFF_{5},\subseteq)$};
\node (4F)  at (6,6) {$  \cup$};
\draw (6,7.5) node[fill=white, inner sep=0.5pt] (4F1)  {\footnotesize $\{1\}$};
\draw[dashed] (2356)--(4F);
\draw[dashed] (4F)--(4F1);
\draw[decorate,thick,mauve,decoration={brace,amplitude=5pt,mirror}] 
    (5.5,5.75)  -- (6.5,5.75)  node [midway,yshift=-0.33cm] {\scriptsize  \color{mauve} $(\FFF_{4},\subseteq)$};
\node (6F)  at (-9,4.5) {$  \cup$};
\draw (-10,6) node[fill=white, inner sep=0.5pt] (6F2) {\footnotesize $\{2\}$}
      (-8,6) node[fill=white, inner sep=0.5pt] (6F3)  {\footnotesize $\{3\}$}
      (-9,7.5) node[fill=white, inner sep=0.5pt] (6F23)  {\footnotesize $\{2,3\}$};
\draw[dashed] (145)--(6F);
\draw[dashed] (6F2)--(6F)--(6F3);
\draw (6F2)--(6F23)--(6F3);
\draw[decorate,thick,mauve,decoration={brace,amplitude=5pt,mirror}] 
    (-10,4.25)  -- (-8,4.25)  node [midway,yshift=-0.33cm] {\scriptsize  \color{mauve} $(\FFF_{6},\subseteq)$};
    
\draw (5F2)--(36);
\draw (5F24)--(346)--(5F14);
\draw (6F2)--(145)--(6F3);
\draw (2356)--(4F1);
\end{tikzpicture}
        \caption{A split graph and the feasible sets posets associeted with its shelling antimatroid.}\label{ex:LattFtot}
\end{figure}
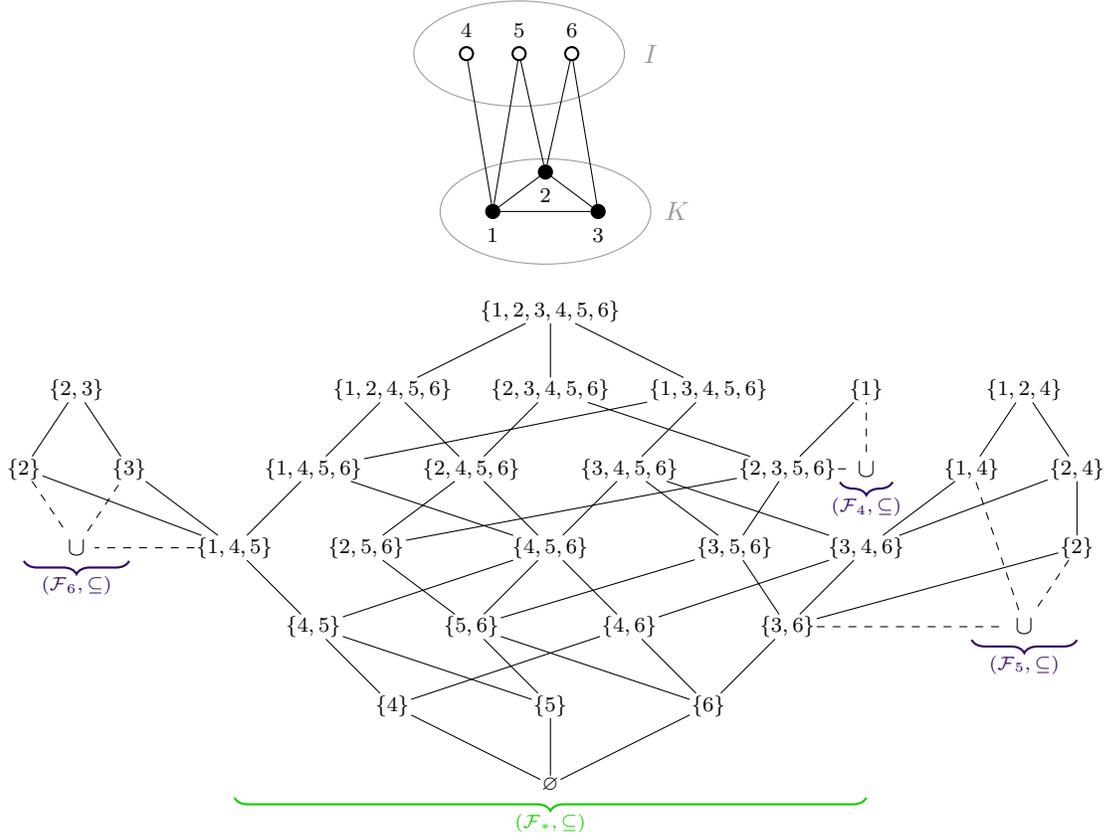

\subsection*{The path poset of a split graph shelling antimatroid}

 For an antimatroid, a \textsl{path} is a feasible set that cannot be decomposed into the union of two other (non-empty) feasible sets.  Alternatively, a path is a feasible set containing a unique element whose removal leaves a feasible set.  When ordered by inclusion, the family of paths forms the path poset. The paths of any antimatroid completely determine the antimatroid. To the contrary, the poset structure of the path poset does not in general determine the antimatroid.  The two following collections of paths, when ordered by inclusion, form isomorphic posets although they arise from two non-isomorphic antimatroids:  $\{ \{a\}$, $\{b\}$, $\{a,c\}$, $\{b,c\} \}$ and $\{ \{d\}$, $\{e\}$, $\{d,f\}$, $\{e,g\} \}$. 

For a split graph shelling antimatroid $(V,\FFF)$ built on a split graph $(K\cup I, E)$, the path poset is easy to obtain in terms of the following sets: 
\begin{align*}
P_1&=\{\{i\}:i\in I\},\\
P_2&=\{ \{\{k\}\cup(N(k)\cap I)\}: k\in K  \},\\
P_3&=\{ \fos(i)\cup \{k\}\cup (N(k)\cap I\setminus\{i\}): i\in I, k\in N(i)  \},\\
P&=P_1\cup P_2\cup P_3.
\end{align*}

\begin{prop}\label{pathSplit}
Given a split graph shelling antimatroid $(V,\FFF)$ built on a split graph $(K\cup I, E)$ without any vertex $i$ in $I$ such that $N(i)=K$, its set of paths equals $P$.
\end{prop}

In the proposition above, the condition forbidding any $i$ in $I$ such that $N(i)=K$ is not very restrictive because if such an $i$ exists, then we change the partition $K\cup I$ into $(K\cup\{i\})\cup (I\setminus\{i\})$.
 
\begin{proof}
Every set in $P$ is feasible, because of Proposition~\ref{Nfclique}. Next we show that every feasible set in $P$ cannot be decomposed into the union of two proper feasible sets. This is trivial for the sets in $P_1$.  So suppose that $F=\{\{k\}\cup(N(k)\cap I)\}$ in $P_2$ is the union of two proper sets $F_1$ and $F_2$ with $k\in F_1$. Then $F_1$ is not feasible because of Proposition~\ref{Nfclique} and the assumptions which ensures that there is no $i$ in $I$ such that $N(i)=K$. Now suppose that $F=\{\fos(i)\cup\{k\}\cup (N(k)\cap I\setminus\{i\})$ in $P_3$ is the union of two proper sets $F_1$ and $F_2$ with $k\in F_1$. Then $F_1$ is not feasible because of Proposition~\ref{Nfclique} and $N(F)\cap I = \{i\}$ and $N(F)\cap K \subseteq N(i)$.

Second, we show that every $F$ in $\FFF$ is the union of some sets in $P$. If $F$ is a $*$-feasible set, we use the sets from $P_1$ and $P_2$. If $F$ is an $i$-feasible set, we use the sets from $P_3$ of the form $\{\fos(i)\cup \{k\}\cup (N(k)\cap I\setminus\{i\})$ with $k\in N(i)$ and some sets from $P_1$. By Corollary~\ref{partitionF} we are done.
\end{proof}

Figure~\ref{figPP1} illustrates Proposition~\ref{pathSplit}. Directly from Proposition~\ref{pathSplit}, we have the following corollary.
 \begin{figure}[h]
        \centering
\begin{tikzpicture}[scale=0.7]
\draw (3.5,0) node[] {{\color{truegray}$K$}};
\draw (3,3) node[] {{\color{truegray}$I$}};
 \draw[truegray] (1,0) ellipse (2cm and 1cm);
 \draw[truegray] (0.5,3) ellipse (2cm and 1cm);
 \draw (0,0) node[circle,draw,fill=black,thick,inner sep=1.75pt,label=below :{\footnotesize $1$}] (1)  {}
 (1,0.75) node[circle,draw,fill=black,thick,inner sep=1.75pt,label=below :{\footnotesize $2$}] (2)  {}
 (2,0) node[circle,draw,fill=black,thick,inner sep=1.75pt,label=below :{\footnotesize $3$}] (3)  {};
 \draw (-0.5,3) node[circle,draw,fill=white,thick,inner sep=1.75pt,label=above :{\footnotesize $4$}] (4)  {}
 (0.5,3) node[circle,draw,fill=white,thick,inner sep=1.75pt,label=above :{\footnotesize $5$}] (5)  {}
 (1.5,3) node[circle,draw,fill=white,thick,inner sep=1.75pt,label=above :{\footnotesize $6$}] (6)  {};
\draw (4)--(1)--(5)--(2)--(6)--(3);
\draw (1)--(2)--(3)--(1);
\end{tikzpicture}
\bigskip
\bigskip
        \begin{tikzpicture}[scale=0.7]
\draw 
      (-3,0) node[fill=white, inner sep=0.5pt] (4) {\footnotesize $\{4\}$}
      (0,0) node[fill=white, inner sep=0.5pt] (5) {\footnotesize $\{5\}$}
      (3,0) node[fill=white, inner sep=0.5pt] (6) {\footnotesize $\{6\}$}
      (-3,1.5) node[fill=white, inner sep=0.5pt] (145) {\footnotesize $\{1,4,5\}$}
      (0,1.5) node[fill=white, inner sep=0.5pt] (256) {\footnotesize $\{2,5,6\}$}
      (3,1.5) node[fill=white, inner sep=0.5pt] (36) {\footnotesize $\{3,6\}$}
      (-6,3) node[fill=white, inner sep=0.5pt] (1245) {\footnotesize $\{1,2,4,5\}$}
      (-2,3) node[fill=white, inner sep=0.5pt] (1345) {\footnotesize $\{1,3,4,5\}$}
      (2,3) node[fill=white, inner sep=0.5pt] (236) {\footnotesize $\{2,3,6\}$}
      (6,3) node[fill=white, inner sep=0.5pt] (1346) {\footnotesize $\{1,3,4,6\}$}
      (0,4.5) node[fill=white, inner sep=0.5pt] (12356) {\footnotesize $\{1,2,3,5,6\}$};
\draw (4)--(145)--(5)--(256)--(6)--(36);
\draw (1245)--(145)--(1345) (36)--(236) (36)--(1346)--(4);
\draw (256)--(12356)--(236);
\draw[rounded corners,blue!75] (-3.5,-0.5) rectangle (3.5,0.5);
\draw[rounded corners,blue!75] (-3.9,1) rectangle (3.7,2);
\draw[rounded corners,blue!75] (-7.15,2.5) rectangle (7.15,5);
\draw (4,-0.25) node[] {{\color{blue}$P_1$}};
\draw (4.2,1.25) node[] {{\color{blue}$P_2$}};
\draw (7.65,2.75) node[] {{\color{blue}$P_3$}};
\end{tikzpicture}
        \caption{A split graph and the path poset associated with its shelling antimatroid.}\label{figPP1}
\end{figure}
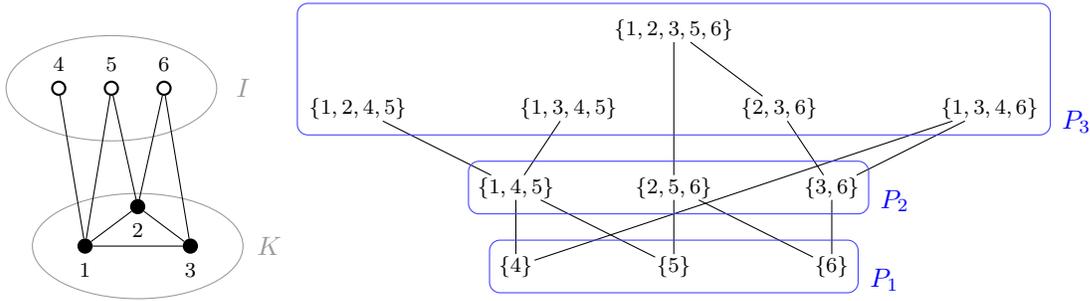

\begin{coroll}
Let $(V,\FFF)$ be a split graph shelling antimatroid built on a split graph $(K\cup I, E)$ without any vertex $i$ in $I$ such that $N(i)=K$. The number of paths in $(V,\FFF)$ is equal to the number of vertices plus the number of edges from $K$ to $I$.
\end{coroll}

\section{Finding a maximum weight feasible set}

\subsection*{Hardness result}
Many classical problems in combinatorial optimization have the following form. For a set system $(V,\FFF)$  and for a function $w:V \rightarrow \R$, find a set $F$ of $\FFF$ maximizing the value of 
\[w(F)=\sum_{f\in F}w(f).\]
For instance, the problem is known to be efficiently solvable for the independent sets of matroids (see \citet{Oxley_2006}) using the greedy algorithm. Since antimatroids capture a combinatorial abstraction of convexity in the same way as matroids capture linear dependence, we investigate the optimization of linear objective functions for antimatroids. It is not known whether a general efficient algorithm exists in the case of antimatroids. Of course, the hardness of finding a maximum weight feasible set depends on the way we encode the antimatroids (see~\citet{Eppstein_13} and~\citet{Enright2000} for more information). If one is given all the feasible sets and a real weight for each element, it is trivial to find a maximum weight feasible set in time polynomial in $|\FFF|$ (in the next section we use a better definition of the size of an antimatroid).

We investigate what happens if we choose a more compact way to encode the information. We use now the path poset to describe an antimatroid. However, optimization on antimatroids given in this compact way is hard as Theorem~\ref{optiAM} shows. We first recall the following theorem due to \citet{Hastad96}, initially stated in terms of a maximum clique.

\begin{thm} \label{CovComp}
There can be no polynomial time algorithm that approximates the problem of finding the maximum size of an independent set in a graph on $n$ vertices to within a factor better than $O(n^{1-\varepsilon})$, for any $\varepsilon> 0$, unless $\mathcal{P} = \mathcal{NP}$.
\end{thm}

For antimatroids given in the form of a path poset, there is a \NPclass -completeness reduction from the maximum independent set problem to the maximum weight feasible set. The reduction given below is an adaptation of a result due to \citet{Eppstein_07}.

\begin{thm}  \label{optiAM}
The problem of finding a maximum weight feasible set in an antimatroid encoded in the form of its path poset is not approximable in polynomial time within a factor better than $O(N^{\frac{1}{2} - \varepsilon})$ for any $\varepsilon > 0$, where $N$ is the number of elements in the path poset, unless $\mathcal{P} = \mathcal{NP}$.
\end{thm}

\begin{proof}
Given any graph $G=(V,E)$ on which we want to find an independent set of maximum size, we define an antimatroid $(A,\FFF)$ by letting $A=V\uplus E$ and defining a feasible set (an element of $\FFF$) as any subset $F$ of $A$ such that $\{v_1,v_2\}\in E\cap F$  then $v_1 \in F$ or $v_2 \in F$. Remark that $(A,\FFF)$ is indeed an antimatroid because it satisfies~\eqref{VinF}, \eqref{Ustabl} and \eqref{accec}.

The path poset of this antimatroid is composed of sets $\{v\}$ for each vertex in $V$ and sets $\{v,e\}$ for each edge $e\in E$ such that $v\in e$. Let $d(v)$ denote the degree of the vertex $v$ and $\delta=0.1$. We define a weight function: $w:A\rightarrow \R$ by setting
\[
w(x)=\left\{
  \begin{array}{cl}
    -d(x)+\delta & \mbox{ if } x\in V \\
    1 & \mbox{ if } x\in E\\
  \end{array}
\right. .
\]

We first show that if $F$ is a feasible set with weight $w(F)$, then we can construct an independent set of $G$ of size at least $w(F)\delta^{-1}$ in polynomial time. To that end, we define a feasible set $F'\subseteq F$ as follows. If $V\cap F$ corresponds to an independent set of vertices in the graph $G$, then $F'= F$. If it is not the case, we select a pair $\{u,v\}\subseteq F$ such that $u\sim v$ in the graph, and remove the element $u$ from $F$. If there is an element $\{u,a\}\in F\cap E$, with $a\notin F$, we also remove $\{u,a\}$ from $F$ (to maintain the feasibility of the set). We repeat this operation until the remaining vertices in the set $F'$ form an independent set in the graph. The remaining elements then form the set $F'$. It is easy to check that $F'$ is always feasible. By the definition of the function $w$, we have the following inequalities, 
\begin{align*}
 w(F)\leq w(F') &\leq  \sum_{v\in V\cap F'}(-d(v)+\delta) + \sum_{e\in E\cap F'} 1 \\
            &\leq  \sum_{v\in V\cap F'}(-d(v)+\delta) + \sum_{v\in V\cap F'} d(v)  = \delta |V\cap F'|.
\end{align*}
So we have an independent set $V\cap F'$ that we can construct in polynomial time with size greater than $w(F)\delta^{-1}$.

Now, let $N$ be the number of paths of $(A,\FFF)$, and suppose we have a $f(N)$-approximation algorithm to find a maximum weight feasible set,  \textit{i.e.} we have an algorithm that returns a feasible set with weight at least $f(N)^{-1}$ times the weight of a maximum  weight feasible set. Assume that  $f(N) \leq O(N^{\frac{1}{2} - \varepsilon})$ for some $0<\varepsilon <1$. We know that $N=|V|+2|E|$, so
\begin{align*}
f(N)\leq O((|V|+2|E|)^{\frac{1}{2} -\varepsilon})\leq O((|V|+|V|^{2})^{\frac{1}{2} -\varepsilon}) 
\leq O(|V|^{1 -\varepsilon'}),
\end{align*}
for a $\varepsilon'\in ]0,1[$. So we have
\[\frac{1}{f(N)} \geq \frac{1}{O(n^{1 - \varepsilon'})}, \]
and we obtain a feasible set with weight at least $\frac{1}{O(n^{1 - \varepsilon'})}$ times the weight $w^*$ of a maximum weight feasible set. By the previous statement, we build an independent set with size at least $(w^{*} / O(n^{1 - \varepsilon'})) \delta^{-1}$, so at least $1 / O(n^{1-\varepsilon'}) $ the size of a maximum independent set, and this contradicts Theorem~\ref{CovComp}. So $f(N) \leq O(N^{\frac{1}{2} - \varepsilon'})$ is impossible  unless $\mathcal{P} = \mathcal{NP}$.
\end{proof}

Moreover, the above theorem remains true also for a subclass of antimatroids (those built in the proof).

\subsection{Optimization on split graph shelling antimatroids}
 
We will now prove that for a weighted split graph shelling antimatroid, the problem of finding a maximum weight feasible set can be done in polynomial time in the size of the input even if the form of the input considered is a more compact representation than the path poset. We use the split graph itself to encode all the information about the feasible sets.

In the case of the poset antimatroids, the optimization problem is solved using the solution to the \emph{maximum closure problem}:
\begin{pbl}\label{MCP}
Given a poset $(V,\leq)$ and a weight function $w: V\rightarrow \R$, find a filter $F$ that maximizes 
\[w(F)=\sum_{f\in F}w(f).\]
\end{pbl}
\citet{Picard_1976} designs a polynomial algorithm to solve Problem~\ref{MCP}, which calls as a subroutine a maximum flow algorithm (\textit{e.g.} \citet{Goldberg1988}). Picard's algorithm runs in $O(mn \log(\frac{n^{2}}{m}))$ time, where $n$ is the number of vertices in $V$ and $m$ the number of cover relations in the poset. Taking advantage of this result, we have the following theorem.

\begin{thm} Giving a split graph $G$ (as a list of vertices and a list of edges), the problem of finding a maximum weight feasible set in the split graph shelling antimatroid defined on $G$ can be done in polynomial time.
\end{thm}

\begin{proof}
We recall that for every split graph shelling antimatroid we introduce a unique poset with relation~$\prec$ (see just before Proposition~\ref{pAMposet0.0}) . The construction of this poset combined to Corollary~\ref{partitionF} and Propositions~\ref{pAMposet0.0} and~\ref{pAMposet0.1} allows us to decompose the problem of finding a maximum feasible set in a split graph antimatroid into several maximum closure problems. Indeed, we first solve the maximum closure problem for $(V,\prec)$, yielding a $*$-feasible set with maximum weight among all the $*$-feasible sets. Then for each $i$ in $I$, we solve the maximum closure problem for $(\ufs(i),\prec)$, yielding a set $S$ such that $S \cup \fos(i)$ is an $i$-feasible set that have maximum weight among all  $i$-feasible sets. The algorithm outputs the feasible set found with maximum weight.
\end{proof}

So suppose that we have a procedure to find a filter in a poset $(V,\leq)$ of maximum weight (given by a function $w$) called $\operatorname{MaxClo}(V,\leq,w)$. In a split graph $(K\cup I, E)$, we look at the element $i$ in $I$ that maximizes the weight of $\fos(i) \cup \operatorname{MaxClo}(\ufs(i),\prec,w)$, we then compare the result with the weight of $\operatorname{MaxClo}(K\cup I,\prec,w)$ and keep the maximum. The time complexity of the algorithm is $O((|I|(|E||K|+|E||I|) \log(\frac{(|K|+|I|)^{2}}{|E|}))$ due to the complexity of $\operatorname{MaxClo}(V,\leq,w)$. Note that if we use a procedure to find a filter in a poset $(V,\leq)$ of minimum weight (given by a function $w$), with very little modifications, our algorithm can be used to return a feasible set of minimum weight.

\section{Free sets and circuits of the split graph shelling antimatroids}

In this last section, the term ``path" takes only its graph-theoretical meaning, while ``circuit" refers to the antimatroidal concept. Our aim is to characterize in simple terms the ``circuits" and ``free sets" of a split graph shelling antimatroid.  Let us first recall some definitions, for a given antimatroid $(V,\FFF)$. The \emph{trace} of $(V,\FFF)$ on a subset $X$ of $V$ is
\[ \tr(\FFF,X) = \{F \cap X : F \in \FFF\}. \]
A subset $X$ of $V$ is \emph{free} if $\tr(\FFF,X) = 2^X$. A \emph{circuit} is a minimal non free subset of $V$. An equivalent characterization reads (for a proof, see for instance~\citet{Korte_Lovasz_Schrader_1991}): a subset $C$ of $V$ is a circuit if and only if $\tr(\FFF,C) = 2^C \setminus \{\{r\}\}$ for some $r$ in $C$ (this element $r$ is unique) .  The element $r$ is then the \emph{root} of $C$, and the pair $(C\setminus\{r\},r)$ is a \emph{rooted circuit}. \citet{Dietrich87} shows that the collection of all rooted circuits determines the initial antimatroid, but the collection of circuits themselves does not always share this property.  She even shows that the collection of `critical' rooted circuits determines the antimatroid, where a rooted circuit $(C\setminus\{r\},r)$ is \emph{critical} when there is no rooted circuit $(D\setminus\{r\},r)$ with the same root $r$ such that the largest feasible set disjoint from $D$ strictly includes the largest feasible set disjoint from $C$.  For a recent reference, see \citet{Nakamura_2013}.

Let $G=(V,E)$ be a chordal graph.  It is known that the rooted circuits of its vertex shelling antimatroid admit the following simple description: a pair $(C,r)$ is a rooted circuit if $C$ consists of two distinct vertices $u$, $v$ such that $r$ is an internal vertex on some chordless path joining $u$ and $v$ (this follows immediately from Corollary~3.4 in \citet{Farber_87}).  Moreover, the circuit $(C,r)$ is critical if and only if the path has exactly three vertices.  For the particular case of split graphs we now provide more efficient characterizations of (critical) circuits, and then of free sets.

\begin{prop}\label{prop_circuits}
Let $(V,\FFF)$ be the vertex shelling antimatroid of the split graph $(K \cup I,E)$.  Set
\begin{align*}
C_1 =& \{ (\{i,j\},k) : k \in K, i,j \in N(k) \cap I \};\\
C_2 =& \{ (\{i,l\},k) : k \in K, i \in N(k) \cap I, 
l \in (N(k) \cap K) \setminus N(i) \};\\
C_3 =& \{ (\{i,j\},k) : k \in K, i \in N(k) \cap I, j \in I \setminus N(k) \text{ and }\\
  & \qquad  \exists m \in K \text{ with } i \not\sim m, j \sim m \}.
\end{align*}
Then the collection of rooted circuits of $(V,\FFF)$ equals $C_1 \cup C_2 \cup C_3$.  Moreover, the collection  of critical rooted circuits equals  $C_1 \cup C_2$.
\end{prop}

\begin{proof}
Notice that any chordless path in a split graph $(K \cup I, E)$  has at most four vertices.  Moreover, if it has three vertices, the internal vertex is in $K$ and at least one extremity is in $I$.  If it has four vertices, the internal vertices are in $K$ and the extremities in $I$.  
The result then follows from the characterization of the circuits of the shelling antimatroid of a chordal graph (which we recall just before the statement): the rooted circuits forming $C_1$ and $C_2$ come from paths with three vertices, those forming $C_3$ come from paths with four vertices.
\end{proof}

\begin{prop}\label{prop_free_sets}
Let $G=(K \cup I,E)$ be a split graph with $L$ and $J$ (possibly empty) subsets of respectively $K$ and $I$.  Then $L \cup J$ is free in the vertex shelling antimatroid of $G$ if and only if either there is no edge between $L$ and $J$, or there exists some vertex $h$ in $J$ such that $L \subseteq N(h)$ and $N(J \setminus\{h\}) \subseteq N(h) \setminus L$.
\end{prop}

\begin{proof}
Assuming first that $X$ is a free set in $(V,\FFF)$, we let $L = X \cap K$ and $J = X \cap I$ (we may have $L$ and/or $J$ empty).  If no edge of $V$ has an extremity in $L$ and the other one in $J$, then $L \cup J$ is as in the first case of the statement.  If there is some edge $\{l,h\}$ with $l \in L$ and $h \in J$, we show that $L$ and $J$ are as in the second case of the statement.  First, there holds $L \subseteq N(h)$ because otherwise for any vertex $u$ in $L \setminus N(h)$, we would find the circuit $\{h,l,u\}$ in $X$ (but a free set cannot contain any circuit).  Second, we prove $N(J \setminus\{h\}) \subseteq N(h) \setminus L$ again by contradiction.  Thus assume some vertex $v$ belongs to $N(J\setminus\{h\}) \setminus (N(h) \setminus L)$.  Then $v$ is adjacent to some $i$ in $J \setminus\{h\}$, and $v$ belongs to either $L$ or $K \setminus N(h)$.  In the first eventuality, $X$ contains the circuit 
$\{h,v,i\}$.  In the second eventuality, whether $l \sim i$ or $l \not\sim i$, the circuit $\{h,l,i\}$ is in $X$.  In both eventualities we reach a contradiction.  Thus $X = L \cup J$ is as in the second case of the statement.

Conversely, assume $L$ and $J$ are as in the statement and let us prove that $X = L \cup J$ contains no circuit, and so that $X$ is free.  If a rooted circuit $(\{i,j\},k)$ from $C_1$ (as in Proposition~\ref{prop_circuits}) is in $X$, our assumption imposes $i=h=j$, a contradiction.  If a rooted circuit $(\{i,l\},k)$ from $C_2$ is in $X$, then our assumption implies first $i=h$ because $l \in N(i) \cap L$, and then $k \notin L$ in contradiction with $k \in X$.  Finally, if a rooted circuit $(\{i,j\},k)$ from $C_3$ is in $X$ with $m$ as in $C_3$, our assumption implies $i=h$, but then $m \in N(k) \setminus N(h)$ is a contradiction with the assumption.  
\end{proof}

\section{Further work}

We studied the structure of split graph shelling antimatroids and described an algorithm to solve the maximum weight feasible set problem in polynomial time. The antimatroids considered form a very special class, but it seems that in general not much is known about the structure of chordal graph shelling antimatroids. We hope our paper will pave the way for further research on chordal graph shelling antimatroids. In particular, the complexity of finding maximum weight feasible sets on chordal graphs shelling antimatroids is an interesting problem.


\nocite{*}
\bibliographystyle{abbrvnat}
\bibliography{splitAM_biblio2}
\label{sec:biblio}

\end{document}